\documentclass[12pt]{amsart}
\usepackage{amssymb,latexsym,amsmath,amscd,amsthm,graphicx, color}
\usepackage[all]{xy}
\usepackage{pgf,tikz}
\usepackage{mathrsfs}
\usepackage[font=small,skip=-0.4 in]{caption}
\usetikzlibrary{arrows}
\definecolor{uuuuuu}{rgb}{0.26666666666666666,0.26666666666666666,0.26666666666666666}
\definecolor{xdxdff}{rgb}{0.49019607843137253,0.49019607843137253,1.}
\definecolor{ffqqqq}{rgb}{1.,0.,0.}

\definecolor{uuuuuu}{rgb}{0.26666666666666666,0.26666666666666666,0.26666666666666666}
\definecolor{qqwuqq}{rgb}{0.,0.39215686274509803,0.}
\definecolor{zzttqq}{rgb}{0.6,0.2,0.}
\definecolor{xdxdff}{rgb}{0.49019607843137253,0.49019607843137253,1.}
\definecolor{qqqqff}{rgb}{0.,0.,1.}
\definecolor{cqcqcq}{rgb}{0.7529411764705882,0.7529411764705882,0.7529411764705882}
\definecolor{sqsqsq}{rgb}{0.12549019607843137,0.12549019607843137,0.12549019607843137}

\theoremstyle{plain}

\newtheorem{theorem}[subsection]{Theorem}

\newtheorem{lemma}[subsection]{Lemma}

\newtheorem{prop}[subsection]{Proposition}

\theoremstyle{definition}
\newtheorem{remark}[subsection]{Remark}

\newtheorem{note}[subsection]{Note}

\newcommand{\uu}{\cup}
\newcommand{\ii}{\cap}
\newcommand{\UU}{\bigcup}
\newcommand{\II}{\bigcap}

\newcommand{\sci}{\subset}
\newcommand{\es}{\emptyset}
\newcommand{\set}[1]{\{#1\}}

\newcommand{\ga}{\alpha}
\newcommand{\gb}{\beta}

\newcommand{\gd}{\delta}
\renewcommand{\gg}{\gamma}

\newcommand{\go}{\omega}

\newcommand{\gs}{\sigma}
\newcommand{\gt}{\tau}

\newcommand{\tit}{\textit}

\newcommand{\C}[1]{\mathcal{#1}}
\newcommand{\D}[1]{\mathbb{#1}}
\newcommand{\F}[1]{\mathfrak{#1}}

\newcommand{\te}{\text}

\newcommand{\nd}{\noindent}

\begin{document}

\nd To appear, Springer Proceedings in Mathematics \& Statistics, \emph{Applied Analysis, Optimization and Soft
Computing,} ICNAAO-2021, Varanasi, India, December 21–23. 
\title{Optimal quantizers for a nonuniform  distribution on a Sierpi\'nski carpet}

\author{ Mrinal Kanti Roychowdhury}
\address{School of Mathematical and Statistical Sciences\\
University of Texas Rio Grande Valley\\
1201 West University Drive\\
Edinburg, TX 78539-2999, USA.}
\email{mrinal.roychowdhury@utrgv.edu}

\thanks{The research of the author was supported by U.S. National Security Agency (NSA) Grant H98230-14-1-0320}

\date{}
\maketitle

\pagestyle{myheadings}\markboth{Mrinal Kanti Roychowdhury}{Optimal quantizers for a nonuniform  distribution on a Sierpi\'nski carpet}

\begin{abstract} The purpose of quantization for a probability distribution is to estimate the probability by a discrete probability with finite support. In this paper, a nonuniform  probability measure $P$ on $\mathbb R^2$ which has support the Sierpi\'nski carpet generated by a set of four contractive similarity mappings with equal similarity ratios has been considered. For this probability measure, the optimal sets of $n$-means and the $n$th quantization errors are investigated for all $n\geq 2$.
\end{abstract}

\section{Introduction}

Quantization is a destructive process. Its purpose is to reduce the cardinality of the representation space, in particular when the input data is real-valued. It is a fundamental problem in signal processing, data compression and information theory. We refer to \cite{GG, GN, Z} for surveys on the subject and comprehensive lists of references to the literature, see also \cite{AW, GKL, GL1, GL2}. Let $\D R^d$ denote the $d$-dimensional Euclidean space, $\|\cdot\|$ denote the Euclidean norm on $\D R^d$ for any $d\geq 1$, and $n\in \D N$. Then, the $n$th \textit{quantization
error} for a Borel probability measure $P$ on $\D R^d$ is defined by
\begin{equation*}  V_n:=V_n(P)=\inf \Big\{\int \min_{a\in\alpha} \|x-a\|^2 dP(x) : \alpha \subset \mathbb R^d, 1\leq \text{card}(\alpha) \leq n \Big\}.\end{equation*}
 If $\int \| x\|^2 dP(x)<\infty$, then there is some set $\alpha$ for
which the infimum is achieved (see \cite{AW, GKL, GL1, GL2}). Such a set $\ga$ for which the infimum occurs and contains no more than $n$ points is called an \tit{optimal set of $n$-means}, or \tit{optimal set of $n$-quantizers}. The collection of all optimal sets of $n$-means for a probability measure $P$ is denoted by $\C C_n:=\C C_n(P)$. It is known that for a continuous probability measure an optimal set of $n$-means always has exactly $n$-elements (see \cite{GL2}).
Given a finite subset $\ga\sci \D R^d$, the Voronoi region generated by $a\in \ga$ is defined by
\[M(a|\ga)=\set{x \in \D R^d : \|x-a\|=\min_{b \in \ga}\|x-b\|}\]
i.e., the Voronoi region generated by $a\in \ga$ is the set of all points in $\D R^d$ which are closest to $a \in \ga$, and the set $\set{M(a|\ga) : a \in \ga}$ is called the \tit{Voronoi diagram} or \tit{Voronoi tessellation} of $\D R^d$ with respect to $\ga$. A Borel measurable partition $\set{A_a : a \in \ga}$ of $\D R^d$  is called a \tit{Voronoi partition} of $\D R^d$ with respect to $\ga$ (and $P$) if $P$-almost surely
$A_a \sci M(a|\ga)$ for every $a \in \ga.$  Given a Voronoi tessellation $\set{M_i}_{i=1}^k$ generated by a set of points $\set{z_i}_{i=1}^k$ (called \tit{sites} or \tit{generators}), the mass centroid  $c_i$ of $M_i$ with respect to the probability measure $P$ is given by
\begin{align*}
c_i=\frac{1}{P(M_i)}\int_{M_i} x dP=\frac{\int_{M_i} x dP}{\int_{M_i} dP}.
\end{align*}
The Voronoi tessellation is called the \tit{centroidal Voronoi tessellation} (CVT) if $z_i=c_i$ for $i=1, 2, \cdots, k$, that is, if the generators are also the centroids of the corresponding Voronoi regions.

Let us now state the following proposition (see \cite{GG, GL2}):
\begin{prop} \label{prop10}
Let $\alpha$ be an optimal set of $n$-means and $a\in \ga$. Then,

$(i)$ $P(M(a|\ga))>0$, $(ii)$ $ P(\partial M(a|\ga))=0$, $(iii)$ $a=E(X : X \in M(a|\ga))$, and $(iv)$ $P$-almost surely the set $\set{M(a|\ga) : a \in \ga}$ forms a Voronoi partition of $\D R^d$.
\end{prop}

Let $\alpha$ be an optimal set of $n$-means and  $a \in \alpha$, then by Proposition~\ref{prop10}, we have
\begin{align*}
a=\frac{1}{P(M(a|\ga))}\int_{M(a|\ga)} x dP=\frac{\int_{M(a|\ga)} x dP}{\int_{M(a|\ga)} dP},
\end{align*}
which implies that $a$ is the centroid of the Voronoi region $M(a|\ga)$ associated with the probability measure $P$ (see also \cite{DFG, R1}).

A transformation $f: X \to X$ on a metric space $(X, d)$ is called \tit{contractive} or a \tit{contraction mapping} if there is a constant $0<c<1$ such that $d(f(x), f(y))\leq c d(x, y)$ for all $x, y \in X$. On the other hand, $f$ is called a \tit{similarity mapping} or a \tit{similitude} if there exists a constant $s>0$ such that $d(f(x), f(y))=s d(x, y)$ for all $x, y\in X$. Here $s$ is called the similarity ratio of the similarity mapping $f$.
Let $C$ be the Cantor set generated by the two contractive similarity mappings $S_1$ and $S_2$ on $\D R$ given by $S_1(x)=r_1 x$ and $ S_2 (x)=r_2 x +(1-r_2)$ where $0<r_1, r_2<1$ and $r_1+r_2<\frac 12$. Let $P=p_1 P\circ S_1^{-1}+p_2 P\circ S_2^{-1}$, where $P\circ S_i^{-1}$ denotes the image measure of $P$ with respect to
$S_i$ for $i=1, 2$ and $(p_1, p_2)$ is a probability vector with $0<p_1, p_2<1$. Then, $P$ is a singular continuous probability measure on $\D R$ with support the Cantor set $C$ (see \cite{H}). For $r_1=r_2=\frac 13$ and $p_1=p_2=\frac 12$,  Graf and Luschgy gave a closed formula to determine the optimal sets of $n$-means for the probability distribution $P$ for any $n\geq 2$ (see \cite{GL3}). For $r_1=\frac 14$, $r_2=\frac 12$, $p_1=\frac 14$ and $p_2=\frac 34$, L. Roychowdhury gave an induction formula to determine the optimal sets of $n$-means and the $n$th quantization error for the probability distribution $P$ for any $n\geq 2 $ (see \cite{R2}).  Let $P$ be a Borel probability measure on $\mathbb R^2$ supported by the Cantor dusts generated by a set of $4^u,\ u\geq 1$, contractive similarity mappings satisfying the strong separation condition. For this probability measure, C\"omez and Roychowdhury determined the optimal sets of $n$-means and the $n$th quantization errors for all $n\geq 2$ (see \cite{CR}). In addition, they showed that though the quantization dimension of the measure $P$ is known, the quantization coefficient for $P$ does not exist.

In this paper, we have considered the probability distribution $P$ given by $P=\frac 1 8P\circ S_1^{-1}+\frac 1 8P\circ S_2^{-1}+\frac 38P\circ S_3^{-1}+\frac 38P\circ S_4^{-1}$ which has support the Sierpi\'nski carpet generated by the four contractive similarity mappings given by $S_1(x_1, x_2)=\frac 13(x_1, x_2)$, $S_2(x_1, x_2)=\frac 13(x_1, x_2) + (\frac 23, 0)$, $S_3(x_1, x_2)=\frac 13(x_1, x_2) +(0, \frac 23)$, and $S_4(x_1, x_2)=\frac 13(x_1, x_2)+(\frac 23, \frac 23)$ for all $(x_1, x_2) \in \D R^2$. The probability distribution $P$ considered in this paper is called `nonuniform' to mean that all the basic squares at a given level that generate the Sierpi\'nski carpet do not have the same probability. For this probability distribution in Proposition~\ref{prop1}, Proposition~\ref{prop2} and Proposition~\ref{prop3}, first we have determined the optimal sets of $n$-means and the $n$th quantization errors for $n=2, 3, \te{ and } 4$. Then, in Theorem~\ref{Th1} we state and prove an induction formula to determine the optimal sets of $n$-means for all $n\geq 2$. We also give some figures to illustrate the locations of the optimal points (see Figure~\ref{Fig1}). In addition, using the induction formula, we obtain some results and observations about the optimal sets of $n$-means which are given in Section~4; a tree diagram of the optimal sets of $n$-means for a certain range of $n$ is also given (see Figure~\ref{Fig2}).

\section{Preliminaries}
In this section, we give the basic definitions and lemmas that will be instrumental in our analysis. For $k\geq 1$, by a word $\go$ of length $k$ over the alphabet $I:=\set{1, 2,3, 4}$ it is meant that $\go:=\go_1\go_2\cdots \go_k$, i.e., $\go$ is a finite sequence of symbols over the alphabet $I$. Here $k$ is called the length of the word $\go$. If $k=0$, i.e., if $\go$ is a word of length zero, we call it the empty word and is denoted by $\es$. Length of a word $\go$ is denoted by $|\go|$. $I^\ast$ denotes the set of all words over the alphabet $I$ including the empty word $\es$. By $\go\gt:=\go_1\cdots \go_k\tau_1\cdots \gt_\ell$ it is meant that the word obtained from the concatenations of the words $\go:=\go_1\go_2\cdots \go_k$ and $\gt:=\gt_1\gt_2\cdots\gt_\ell$ for $k,\ell\geq 0$. The maps $S_i :\D R^2 \to \D R^2,\ 1\leq i \leq 4, $ will be the generating maps of the Sierpi\'nski carpet defined as before.
For $\go=\go_1\go_2 \cdots\go_k \in I^k$, set $S_\go=S_{\go_1}\circ \cdots \circ S_{\go_k}$
and $J_\go=S_{\go}([0, 1]\times [0, 1])$. For the empty word $\emptyset $, by $S_{\emptyset}$ we mean the identity mapping on $\D R^2$, and write $J=J_{\emptyset}=S_{\emptyset}([0,1]\times [0, 1])=[0, 1]\times [0, 1]$.  The sets $\{J_\go : \go \in \{1, 2, 3, 4 \}^k \}$ are just the $4^k$ squares in the $k$th level in the construction of the Sierpi\'nski carpet. The squares $J_{\go 1}$, $J_{\go 2}$, $J_{\go 3}$ and $J_{\go 4}$ into which $J_\go$ is split up at the $(k+1)$th level are called the basic squares of $J_\go$. The set $S=\cap_{k \in \D N} \cup_{\go \in \{1, 2, 3, 4 \}^k} J_\go$ is the Sierpi\'nski carpet and equals the support of the probability measure $P$  given by $P =\frac 1 8 P \circ S_1^{-1} + \frac 1 8 P\circ S_2^{-1} +\frac 38  P\circ S_3^{-1}+\frac 38 P\circ S_4^{-1}$. Set $s_1=s_2=s_3=s_4=\frac 13$, $p_1=p_2=\frac 18$ and $p_3=p_4=\frac 38$, and for $\go=\go_1 \go_2 \cdots \go_k \in I^k$, write
$c(\go):=\te{card}(\{i : \go_i=3 \te{ or } 4, \, 1\leq i\leq k\})$, where $\te{card}(A)$ of a set $A$ represents the number of elements in the set $A$. Then, for $\go=\go_1\go_2 \cdots\go_k \in I^k$, $k\geq 1$, we have
\[s_\go=\frac 1 {3^k} \te{ and } p_\go=p_{\go_1}p_{\go_2}\cdots p_{\go_k}=\frac{3^{c(\go)}}{8^k}.\]

Let us now give the following lemma.

\begin{lemma} \label{lemma1} Let $f: \D R \to \D R^+$ be Borel measurable and $k\in \D N$. Then,
\[\int f \,dP=\sum_{\go \in I^k} p_\go\int f\circ S_\go \,dP.\]
\end{lemma}
\begin{proof}
We know $P =p_1 P \circ S_1^{-1} + p_2 P\circ S_2^{-1} +p_3 P\circ S_3^{-1}+p_4 P\circ S_4^{-1}$, and so by induction $P=\sum_{\go \in I^k} p_\go P\circ S_\go^{-1}$, and thus the lemma is yielded.
\end{proof}
Let $S_{(i1)}, \, S_{(i2)}$ be the horizontal and vertical components of the transformation $S_i$ for $i=1, 2, 3, 4$. Then, for any $(x_1, x_2) \in \D R^2$ we have
$S_{(11)}(x_1) =\frac 1 3 x_1$, $ S_{(12)}(x_2)=\frac 1 3 x_2$, $S_{(21)}(x_1)=\frac 1 3 x_1 +\frac 23$, $S_{(22)}(x_2)=\frac 1 3 x_2$, $S_{(31)}(x_1)=\frac 1 3 x_1$, $S_{(32)}(x_2)= \frac 1 3 x_2+ \frac 2 3$, and $S_{(41)}(x_1)=\frac 1 3 x_1 +\frac 23$, $S_{(42)}(x_2)= \frac 1 3 x_2+ \frac 2 3$. Let $X:=(X_1, X_2)$ be a bivariate continuous random variable with distribution $P$. Let $P_1, P_2$ be the marginal distributions of $P$, i.e., $P_1(A)=P(A\times \D R)$ for all $A \in \F B$, and $P_2(B)=P(\D R \times B)$ for all $B \in \F B$. Here $\F B$ is the Borel $\gs$-algebra on $\D R$. Then, $X_1$ has distribution $P_1$ and $X_2$ has distribution $P_2$.

Let us now state the  following lemma. The proof is similar to Lemma~2.2 in \cite{CR}.

\begin{lemma} \label{lemma2222} Let $P_1$ and $P_2$ be the marginal distributions of the probability measure $P$. Then,
\begin{itemize}
\item[]
$P_1 =\frac 1{8} P_1 \circ S_{(11)}^{-1} + \frac 1{8} P_1\circ S_{(21)}^{-1} +\frac {3}{8}  P_1\circ S_{(31)}^{-1}+\frac {3}{8}  P_1\circ S_{(41)}^{-1}$  and
\item[]
$P_2 =\frac 1 {8} P_2 \circ S_{(12)}^{-1} + \frac 1{8} P_2\circ S_{(22)}^{-1} +\frac {3}{8}  P_2\circ S_{(32)}^{-1}+\frac 38  P_2\circ S_{(42)}^{-1}$.
\end{itemize}
\end{lemma}

Let us now give the following lemma.
\begin{lemma} \label{lemma333} Let $E(X)$ and $V(X)$ denote the the expected vector and the expected squared distance of the random variable $X$. Then, \[E(X)=(E(X_1), \, E(X_2))=(\frac 12, \frac 34) \te{ and } V:=V(X)=E\|X-(\frac 12, \frac 34)\|^2=\frac 7{32}.\]
\end{lemma}

\begin{proof} We have
\begin{align*}
&E(X_1)=\int x \, dP_1=\frac 1 8 \int x\, dP_1 \circ S_{(11)}^{-1} + \frac 1 8 \int x\, dP_1\circ S_{(21)}^{-1}+ \frac 38 \int x\, dP_1\circ S_{(31)}^{-1} +\frac 38  \int x\, dP_1\circ S_{(41)}^{-1}\\
&=\frac 1 8\int \frac 13 x\, dP_1 + \frac 18 \int (\frac 1 3 x+\frac  23)\, dP_1 +\frac 38 \int \frac 13 x\, dP_1 + \frac 38\int (\frac 1 3 x+\frac  23)\, dP_1,
\end{align*}
which after simplification yields $E(X_1)=\frac 12$, and similarly $E(X_2)=\frac 34$. Now,
\begin{align*}
&E(X_1^2)=\int x^2 \, dP_1\\
&=\frac 1 8 \int x^2\, dP_1 \circ S_{(11)}^{-1} + \frac 1 8 \int x^2\, dP_1\circ S_{(21)}^{-1} +\frac 3 8 \int x^2\, dP_1\circ S_{(31)}^{-1} +\frac 38  \int x^2\, dP_1\circ S_{(41)}^{-1}\\
&=\frac 1 8 \int (\frac 13 x)^2\, dP_1 + \frac 1 8\int (\frac 1 3 x+\frac  23)^2\, dP_1 +\frac 38 \int (\frac 13 x)^2\, dP_1 + \frac 38\int (\frac 1 3 x+\frac  23)^2\, dP_1\\
&=\frac 1 2 \int \frac 19 x^2\, dP_1 + \frac 12\int (\frac 1 9 x^2+ \frac 4 9 x + \frac  49)\, dP_1\\
&=\frac 1 {18} E(X_1^2)+\frac 1 {18} E(X_1^2)+\frac 4{18} E(X_1) +\frac {4}{18}\\
&=\frac 1 9 E(X_1^2)+\frac 1 3.
\end{align*}
This implies  $E(X_1^2)=\frac 38$. Similarly, we can show that $E(X_2^2)=\frac {21}{32}$. Thus,
$V(X_1)=E(X_1^2)-(E(X_1))^2=\frac 38  - \frac 14=\frac 18,$
and similarly $V(X_2)=\frac 3{32}$. Hence,
 \begin{align*} & E\|X-(\frac 12, \frac 34)\|^2=E(X_1-\frac 12)^2 +E(X_2-\frac 34)^2=V(X_1)+V(X_2)=\frac 7{32}.
\end{align*}
Thus, the proof of the lemma follows.
\end{proof}

Let us now give the following note.

\begin{note} From Lemma~\ref{lemma333}  it follows that the optimal set of one-mean is the expected vector and the corresponding quantization error is the expected squared distance of the random variable $X$. For words $\gb, \gg, \cdots, \gd$ in $I^\ast$, by $a(\gb, \gg, \cdots, \gd)$ we mean the conditional expected vector of the random variable $X$ given $J_\gb\uu J_\gg \uu\cdots \uu J_\gd,$ i.e.,
\begin{equation} \label{eq0}a(\gb, \gg, \cdots, \gd)=E(X|X\in J_\gb \uu J_\gg \uu \cdots \uu J_\gd)=\frac{1}{P(J_\gb\uu \cdots \uu J_\gd)}\int_{J_\gb\uu \cdots \uu J_\gd} x dP.
\end{equation}
For $\go \in I^k$, $k\geq 1$, since $a(\go)=E(X : X \in J_\go)$, using Lemma~\ref{lemma1}, we have
\begin{align*}
&a(\go)=\frac{1}{P(J_\go)} \int_{J_\go} x \,dP(x)=\int_{J_\go} x\, dP\circ S_\go^{-1}(x)=\int S_\go(x)\, dP(x)=E(S_\go(X))=S_\go(\frac 12, \frac 34).
\end{align*}  For any $(a, b) \in \D R^2$, $E\|X-(a, b)\|^2=V+\|(\frac 12, \frac 34)-(a, b)\|^2.$
In fact, for any $\go \in I^k$, $k\geq 1$, we have
$\int_{J_\go}\|x-(a, b)\|^2 dP= p_\go \int\|(x_1, x_2) -(a, b)\|^2 dP\circ S_\go^{-1},$
which implies
\begin{equation} \label{eq1}
\int_{J_\go}\|x-(a, b)\|^2 dP=p_\go \Big(s_\go^2V+\|a(\go)-(a, b)\|^2\Big).
\end{equation}
The expressions \eqref{eq0} and \eqref{eq1}  are useful to obtain the optimal sets and the corresponding quantization errors with respect to the probability distribution $P$. The Sierpi\'nski carpet has the maximum symmetry with respect to the vertical line $x_1=\frac 12$, i.e., with respect to the line $x_1=\frac 12$ the Sierpi\'nski carpet is geometrically symmetric as well as symmetric with respect to the probability distribution: if the two basic rectangles of similar geometrical shape lie in the opposite sides of the line $x_1=\frac 12$, and are equidistant from the line $x_1=\frac 12$, then they have the same probability.
\end{note}

\begin{figure}
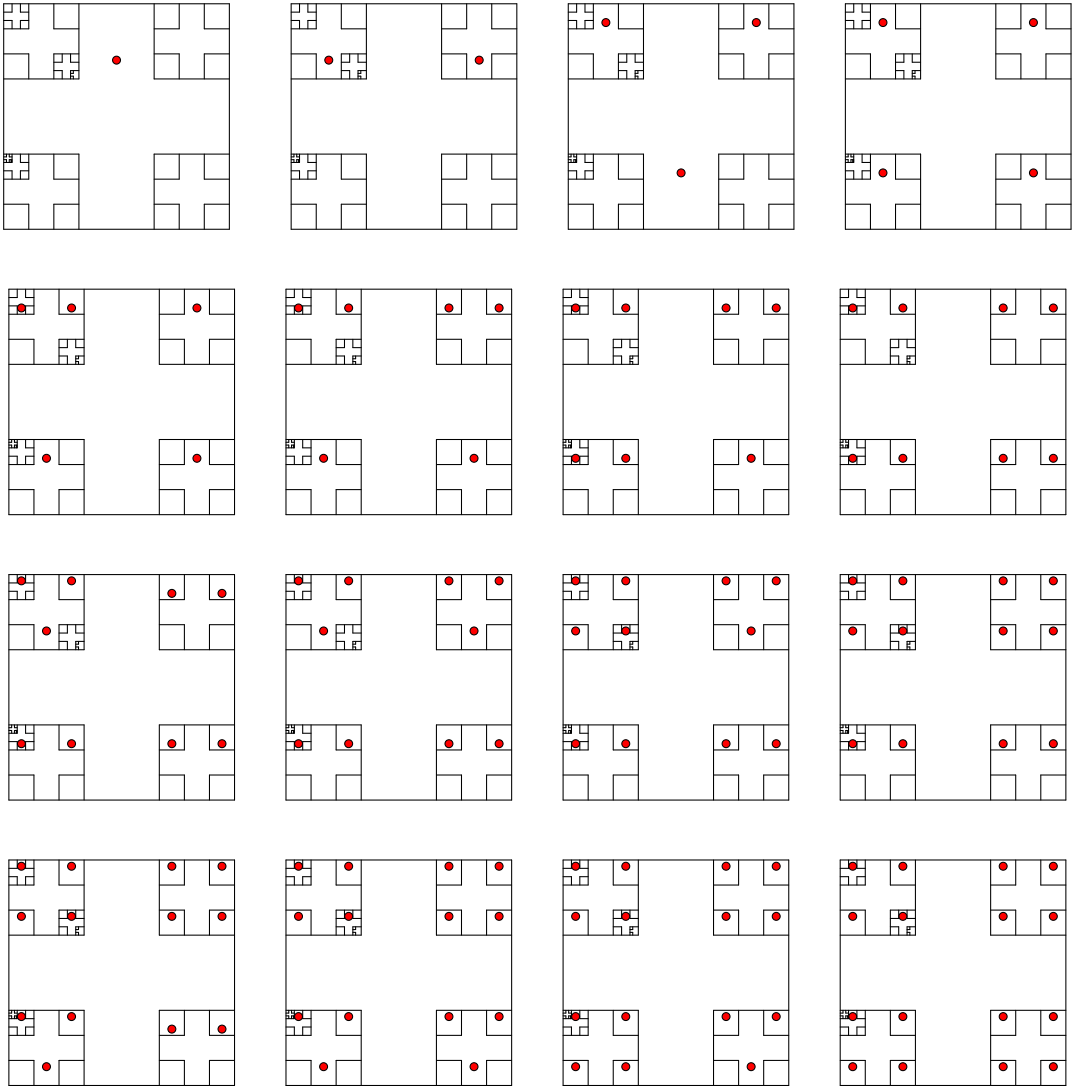


\vspace{ 0.6 in}
\caption{Configuration of the points in an optimal set of $n$-means for $1\leq n\leq 16$.}  \label{Fig1}
\end{figure}

\section{Optimal sets of $n$-means for all $n\geq 2$}

In this section we determine the optimal sets of $n$-means for all $n\geq 2$. First, prove the following proposition.

\begin{prop}\label{prop1} The set $\ga=\set{a(1, 3), a(2, 4)}$, where  $a(1, 3)=(\frac{1}{6}, \frac{3}{4})$ and $a(2,4)=(\frac{5}{6}, \frac{3}{4})$, is an optimal set of two-means with quantization error $V_2=\frac{31}{288}=0.107639$.
\end{prop}
\begin{proof}
Since the Sierpi\'nski carpet has the maximum symmetry with respect to the vertical line $x_1=\frac 12$, among all the pairs of two points which have the boundaries of the Voronoi regions oblique lines passing through the point $(\frac 12, \frac 34)$, the two points which have the boundary of the Voronoi regions the line $x_1=\frac 12$ will give the smallest distortion error. Again, we know that the two points which give the smallest distortion error are the centroids of their own Voronoi regions. Let $(a_1, b_1)$ and $(a_2, b_2)$ be the centroids of the left half and the right half of the Sierpi\'nski carpet with respect to the line   $x_1=\frac 12$, respectively. Then using \eqref{eq0}, we have
\begin{align*}
(a_1, b_1)=E(X : X\in J_1\uu J_3)=\frac{1}{P(J_1\uu J_3)}\int_{J_1\uu J_3}x dP=(\frac{1}{6},\frac{3}{4})
\end{align*}
and
\begin{align*}
(a_2, b_2)=E(X : X\in J_2\uu J_4)=\frac{1}{P(J_2\uu J_4)}\int_{J_2\uu J_4} x dP=(\frac{5}{6},\frac{3}{4}).
\end{align*}
Write $\ga:=\set{(\frac{1}{6},\frac{3}{4}), (\frac{5}{6},\frac{3}{4})}$. Then, the distortion error is obtained as
\[\int\min_{c\in \ga}\|x-c\|^2 dP=\mathop{\int}\limits_{J_1\uu J_3} \|x-(\frac{1}{6},\frac{3}{4})\|^2 dP+\mathop{\int}\limits_{J_2\uu J_4} \|x-(\frac{5}{6},\frac{3}{4})\|^2 dP=\frac{31}{288}=0.107639.\]
Since $V_2$ is the quantization error for two-means, we have $0.107639\geq V_2$. We now show that the points in an optimal set of two means cannot lie on a vertical line. Suppose that the points in an optimal set of two-means lie on a vertical line. Then, we can assume that $\gb=\set{(p, a), (p, b)}$ is an optimal set of two-means with $a\leq b$.  Then, by the properties of centroids we have
\[(p, a) P(M((p, a)|\gb))+(p, b) P(M((p, b)|\gb))=(\frac 12, \frac 34),\]
which implies $p P(M((p, a)|\gb))+p P(M((p, b)|\gb))=\frac 12$ and  $a P(M((p, a)|\gb))+b P(M((p, b)|\gb))=\frac 34$. Thus, we see that $p=\frac 12 $, and the two points $(p, a)$ and $(p, b)$ lie on the opposite sides of the point $(\frac 12, \frac 34)$. Since the optimal points are the centroids of their own Voronoi regions, we have $0\leq a\leq \frac 34\leq  b\leq 1$ implying $\frac 12(0+\frac 34)=\frac 38\leq \frac 12(a+b)\leq \frac 12(\frac 34+1)=\frac 78<\frac 89$, and so $J_{3 3} \uu J_{3 4}\uu J_{43}\uu J_{44} \sci M((\frac 12, b)|\gb)$ and $J_1\uu J_2\sci M((\frac 12, a)|\gb)$.  Suppose that $a\geq \frac 5{12}$. Then as $a(33,34, 43, 44)=E(X : X\in J_{3 3} \uu J_{3 4}\uu J_{43}\uu J_{44})=(\frac{1}{2},\frac{35}{36})$, we have
\[\int\min_{c\in \ga}\|x-c\|^2 dP\geq \mathop{\int}\limits_{J_1\uu J_2} \|x-(\frac{1}{2},\frac 5{12})\|^2 dP+\mathop{\int}\limits_{J_{3 3} \uu J_{3 4}\uu J_{43}\uu J_{44}} \|x-(\frac{1}{2},\frac{35}{36})\|^2 dP=\frac{515}{4608}=0.111762,\]
which is a contradiction, as $0.111762>0.107639\geq V_2$ and $\ga$ is an optimal set of two-means. Thus, we can assume that $a<\frac 5{12}$. Since $a<\frac 5{12}$ and $b\leq 1$, we have $\frac 12(a+b)\leq \frac 12(\frac 5{12}+1)=\frac {17}{24}$, which yields that $B \sci M((\frac 12, b)|\ga)$ where $B=J_{3 3} \uu J_{3 4}\uu J_{43}\uu J_{44}\uu J_{313}\uu J_{314}\uu J_{323}\uu J_{324}\uu J_{413}\uu J_{414}\uu J_{423}\uu J_{424}$. Using \eqref{eq0}, we have $E(X : X\in B)=(\frac{1}{2},\frac{503}{540})$ which implies that $b\leq \frac{503}{540}$. Now if $a\geq \frac 13$, we have
\[\int\min_{c\in \ga}\|x-c\|^2 dP\geq \mathop{\int}\limits_{J_1\uu J_2} \|x-(\frac{1}{2},\frac 13)\|^2 dP+\mathop{\int}\limits_{B} \|x-(\frac{1}{2}, \frac{503}{540})\|^2 dP=\frac{106847}{829440}=0.128818>V_2,\]
which is a contradiction. So, we can assume that $a<\frac 13$. Then, $J_1\uu J_1\sci M((\frac 12, a)|\ga)$ and $J_3\uu J_4\sci M((\frac 12, b)|\ga)$, and so
$(\frac 12, a)=E(X : X\in J_1\uu J_2)=(\frac{1}{2},\frac{1}{4})$ and $(\frac 12, b)=E(X : X \in J_3\uu J_4)=(\frac{1}{2},\frac{11}{12})$, and
\[\int\min_{c\in \ga}\|x-c\|^2 dP= \mathop{\int}\limits_{J_1\uu J_2} \|x-(\frac{1}{2},\frac 14)\|^2 dP+\mathop{\int}\limits_{J_3\uu J_4} \|x-(\frac{1}{2},\frac{11}{12})\|^2 dP=\frac{13}{96}=0.135417>V_2,\]
which leads to another contradiction. Therefore, we can assume that the points in an optimal set of two-means cannot lie on a vertical line. Hence, $\ga=\set{(\frac{1}{6},\frac{3}{4}), (\frac{5}{6},\frac{3}{4})}$ forms an optimal set of two-means with quantization error $V_2=\frac{31}{288}=0.107639$.
\end{proof}

\begin{remark}
The set $\ga$ in Proposition~\ref{prop1} forms a unique optimal set of two-means.
\end{remark}

\begin{prop}\label{prop2}
 The set $\ga=\set{a(1, 2), a(3), a(4)}$, where $a(1, 2)=E(X : X \in J_1\uu J_2)=(\frac{1}{2},\frac{1}{4})$, $a(3)=E(X : X\in J_3)=(\frac{1}{6},\frac{11}{12})$ and $a(4)=E(X : X \in J_4)=(\frac{5}{6},\frac{11}{12})$, forms an optimal set of three-means with quantization error $V_3=\frac{5}{96}=0.0520833$.

\end{prop}

\begin{proof}
Let us first consider the three-point set $\gb$ given by $\gb=\set{a(1, 2), a(3), a(4)} $. Then, the distortion error is obtained as
\begin{align*}
&\int\min_{c\in \ga}\|x-c\|^2 dP\\
&=\int_{J_1\uu J_2} \|x-a(1, 2)\|^2 dP+ \int_{J_3} \|x-a(3)\|^2 dP+\int_{J_4} \|x-a(4)\|^2 dP=0.0520833.
\end{align*}
Since $V_3$ is the quantization error for an optimal set of three-means, we have $0.0520833\geq V_3$. Let $\ga:=\set{(a_i, b_i) : 1\leq i\leq 3}$ be an optimal set of three-means. Since the optimal points are the centroids of their own Voronoi regions, we have $\ga\sci [0, 1]\times [0, 1]$. Then, by the definition of centroid, we have
\[\sum_{(a_i, b_i)\in \ga} (a_i, b_i) P(M((a_i, b_i)|\ga)) =(\frac 12, \frac 34),\]
which implies $\sum_{(a_i, b_i)\in \ga} a_i P(M((a_i, b_i)|\ga)) =\frac 12$ and $\sum_{(a_i, b_i)\in \ga} b_i P(M((a_i, b_i)|\ga)) =\frac 34$. Thus, we conclude that all the points in an optimal set cannot lie in one side of the vertical line $x_1=\frac 12$ or in one side of the horizontal line $x_2=\frac 34$. Without any loss of generality, due to symmetry we can assume that one of the optimal points, say $(a_1, b_1)$, lies on the vertical line $x_1=\frac 12$, i.e., $a_1=\frac 12$, and the optimal points $(a_2, b_2)$ and $(a_3, b_3)$ lie on a horizontal line and are equidistant from the vertical line $x_1=\frac 12$. Further, due to symmetry we can assume that $(a_2, b_2)$ and $(a_3, b_3)$ lie on the vertical lines $x_1=\frac 16$ and $x_1=\frac 56$ respectively, i.e., $a_2=\frac 16$ and $a_3=\frac 56$.

Suppose that $(\frac 12, b_1)$ lies on or above the horizontal line $x_2=\frac 34$, and so $(\frac 16, b_2)$ and $(\frac 56 , b_3)$ lie on or below the line $x_2=\frac 34$. Then, if $\frac 23\leq b_2, b_3\leq \frac 34$, we have
\begin{align*}
&\int\min_{c\in \ga}\|x-c\|^2 dP\geq 2\mathop{ \int}\limits_{J_1\uu J_{31}\uu J_{33}} \min_{\frac{2}{3}\leq b\leq \frac{3}{4}}\|x-(\frac 16, b)\|^2 dP=0.0820313>V_3,
\end{align*}
which is a contradiction. If $\frac 12\leq b_2, b_3\leq \frac 23$,
\begin{align*}
&\int\min_{c\in \ga}\|x-c\|^2 dP\\
&\geq 2\Big(\mathop{ \int}\limits_{J_1\uu J_{31}\uu J_{321}} \min_{\frac{1}{2}\leq b\leq \frac{2}{3}}\|x-(\frac 16, b)\|^2 dP+\mathop{\int}\limits_{J_{33}}\|x-(\frac 16, \frac 23)\|^2 dP+\mathop{ \int}\limits_{J_{342}\uu J_{344}} \min_{\frac{3}{4}\leq b\leq 1}\|x-(\frac 12, b)\|^2 dP\Big)\\
&=2 \Big(\frac{6521}{442368}+\frac{281}{18432}+\frac{277}{110592}\Big)=0.0649821>V_3,
\end{align*}
which leads to a contradiction.
If $\frac 13\leq b_2, b_3\leq \frac 12$, then
\begin{align*}
&\int\min_{c\in \ga}\|x-c\|^2 dP\\
&\geq 2\Big(\mathop{ \int}\limits_{J_{31}\uu J_{321}\uu J_{331}} \|x-(\frac 16, \frac 12)\|^2 dP+\mathop{\int}\limits_{J_1}\|x-(\frac 16, \frac 13)\|^2 dP+\mathop{ \int}\limits_{J_{34}\uu J_{334}} \min_{\frac{3}{4}\leq b\leq 1}\|x-(\frac 12, b)\|^2 dP\Big)\\
&=2 \Big(\frac{811}{110592}+\frac{1}{256}+\frac{78373}{4866048}\Big)=0.0546912>V_3,
\end{align*}
which gives a contradiction. If $0\leq  b_2, b_3\leq \frac 13$, then
\begin{align*}
&\int\min_{c\in \ga}\|x-c\|^2 dP\geq 2\Big(\mathop{ \int}\limits_{J_1} \|x-a(1)\|^2 dP+\mathop{ \int}\limits_{J_{33}\uu J_{34}} \min_{\frac{3}{4}\leq b\leq 1}\|x-(\frac 12, b)\|^2 dP\Big)\\
&=2 \Big(\frac{7}{2304}+\frac{109}{3072}\Big)=0.0770399>V_3
\end{align*}
which leads to another contradiction. Therefore, we can assume that $(\frac 12, b_1)$ lies on or below the horizontal line $x_2=\frac 34$, and  $(\frac 16, b_2)$ and $(\frac 56 , b_3)$ lie on or above the line $x_2=\frac 34$. Notice that for any position of $(\frac 12, b_1)$ on or below the line $x_2=\frac 34$, always $J_{31}\uu J_{33}\uu J_{34} \sci M((\frac 16, b_2)|\ga)$ which implies that $b_2\leq \frac{79}{84}$. Similarly, $b_3\leq \frac{79}{84}$. Suppose that $\frac 12\leq b_1\leq \frac 34$. Then, writing $A=J_{133}\uu J_{321}\uu J_{324}$ and  $B=J_{11}\uu J_{12}\uu J_{14}\uu J_{132}$, we have
\begin{align*}
&\int\min_{c\in \ga}\|x-c\|^2 dP\\
&\geq 2\Big(\mathop{ \int}\limits_{J_{31}\uu J_{33}\uu J_{34} \uu J_{323}}\min_{\frac{3}{4}\leq b\leq \frac{79}{84}} \|x-(\frac 16, b)\|^2 dP+\mathop{ \int}\limits_{A} \|x-(\frac 16, \frac 34)\|^2 dP+\mathop{ \int}\limits_{B} \|x-(\frac 12, \frac 12)\|^2 dP\Big)\\
&=2 \Big(\frac{588517}{78299136}+\frac{5347}{1327104}+\frac{6601}{442368}\Big)=0.0529346>V_3,
\end{align*}
which is a contradiction. So, we can assume that $b_1< \frac 12$. Suppose that $\frac 13\leq b_1< \frac 12$. Then, as $\frac 34\leq b_2\leq \frac{79}{84}$, we see that $J_{31}\uu J_{33}\uu J_{34} \uu J_{321}\uu J_{323}\uu J_{324}\sci M((\frac 16, b_2)|\ga)$.
Then, writing $A_1:=J_{31}\uu J_{33}\uu J_{34} \uu J_{321}\uu J_{323}\uu J_{324}$ and $A_2:=J_{322}\uu J_{1331}\uu J_{1333}\uu J_{1334}\uu J_{13323}\uu J_{13324}$ and  $A_3:=J_{11}\uu J_{12}\uu J_{14}\uu J_{131}\uu J_{132}\uu J_{134}\uu J_{13322}$, we have
\begin{align*}
&\int\min_{c\in \ga}\|x-c\|^2 dP\\
&\geq 2\Big(\mathop{ \int}\limits_{A_1}\min_{\frac{3}{4}\leq b\leq \frac{79}{84}} \|x-(\frac 16, b)\|^2 dP+\mathop{ \int}\limits_{A_2} \|x-(\frac 16, \frac 34)\|^2 dP+\mathop{ \int}\limits_{A_3} \|x-(\frac 12, \frac 13)\|^2 dP\Big)\\
&=2 \Big(\frac{242191}{27869184}+\frac{4135547}{1146617856}+\frac{31584803}{2293235712}\Big)=0.0521401>V_3,
\end{align*}
which gives a contradiction. So, we can assume that $b_1\leq \frac 13$. Then, notice that $J_{11}\uu J_{12}\uu J_{132}\uu J_{141}\uu J_{142}\uu J_{144}\uu J_{21}\uu J_{22}\uu J_{241}\uu J_{231}\uu J_{232}\uu J_{233}\sci M((\frac 12, b_1)|\ga)$ which implies that $b_1\geq \frac{13}{68}$. Thus, we have $\frac{13}{68}\leq b_1\leq \frac 13$. Suppose that $\frac 34\leq b_2, b_3\leq \frac 56$.  Then,
\begin{align*}
&\int\min_{c\in \ga}\|x-c\|^2 dP\\
&\geq 2\Big(\mathop{ \int}\limits_{J_{3}}\min_{\frac{3}{4}\leq b\leq \frac 56} \|x-(\frac 16, b)\|^2 dP+\mathop{ \int}\limits_{J_{11}\uu J_{12}\uu J_{14}\uu J_{131}\uu J_{132}} \min_{\frac{13}{68}\leq b\leq \frac{1}{3}} \|x-(\frac 12, b)\|^2 dP\\
&+\mathop{ \int}\limits_{J_{1331}\uu J_{1333}\uu J_{1334}} \|x-(\frac 16, \frac 34)\|^2 dP+\mathop{ \int}\limits_{J_{134}} \|x-(\frac 12, \frac 13)\|^2 dP\Big)\\
&=2 \Big(\frac{3}{256}+ \frac{147359}{15261696}+\frac{32969}{10616832}+ \frac{3881}{1327104}\Big)=0.054808>V_3,
\end{align*}
which leads to a contradiction. So, we can assume that $\frac 56<b_2, b_3\leq 1$. Then, we have $J_1\uu J_2\sci M((\frac 12, b_1)|\ga)$, $J_3\sci M((\frac 16, b_2)|\ga)$ and $J_4\sci M((\frac 56, b_3)|\ga)$ which yield that $(\frac 12, b_1)=a(1, 2)$, $(\frac 16, b_2)=a(3)$ and $(\frac 56, b_3)=a(4)$, and the quantization error is $V_3=\frac{5}{96}=0.0520833$. Thus, the proof of the proposition is complete.

\end{proof}

\begin{prop} \label{prop3}  The set $\ga=\set{a(1), a(2), a(3), a(4)}$ forms an optimal set of four-means with quantization error $V_4=\frac{7}{288}=0.0243056$.
\end{prop}

\begin{proof}
Let us consider the four-point set $\gb$ given by $\gb:=\set{a(1), a(2), a(3), a(4)}$. Then, the distortion error is given by
\[\int\min_{c\in \gb}\|x-c\|^2 dP=\sum_{i=1}^4 \int_{J_i}\|x-a(i)\|^2 dP=\frac{7}{288}=0.0243056.\]
Since, $V_4$ is the quantization error for four-means, we have $0.0243056\geq V_4$. As the optimal points are the centroids of their own Voronoi regions, $\ga \sci J$. Let $\ga$ be an optimal set of $n$-means for $n=4$.
By the definition of centroid, we know
\begin{equation} \label{eq100}
\sum_{(a, b) \in \ga} (a, b) P(M((a, b)|\ga))=(\frac 12, \frac 34).
\end{equation}
If all the points of $\ga$ are below the line $x_2=\frac 34$, i.e., if $b<\frac 34$ for all $(a, b)\in \ga$, then by \eqref{eq100}, we see that $\frac 34=\sum_{(a, b) \in \ga} b P(M((a, b)|\ga))<\sum_{(a, b) \in \ga} \frac 34 P(M((a, b)|\ga))=\frac 34$, which is a contradiction. Similarly, it follows that if all the points of $\ga$ are above the line $x_2=\frac 34$, or left of the line $x_1=\frac 12$, or right of the line $x_1=\frac 12$, a contradiction will arise.
Suppose that all the points of $\ga$ are on the line $x_2=\frac 34 $. Then, for $(x_1, x_2) \in\uu_{i, j=3}^4 J_{ij}$, we have $\min_{c\in \ga}\|(x_1,x_2)-c\|\geq \frac{5}{36}$, and for $(x_1, x_2) \in\uu_{i, j=1}^2 J_{ij}$, we have $\min_{c\in \ga}\|(x_1,x_2)-c\|\geq \frac{23}{36}$, which implies that
\begin{align*}
&\int \min_{c\in \ga}\|x-c\|^2 dP\geq 4 \mathop{\int}\limits_{J_{33}}\min_{c\in \ga}\|(x_1,x_2)-c\|^2dP+4 \mathop{\int}\limits_{J_{11}}\min_{c\in \ga}\|(x_1,x_2)-c\|^2dP\\
&=4 \Big(\frac 5{36}\Big)^2 P(J_{33})+4 \Big(\frac {23}{36}\Big)^2 P(J_{11})=\frac{377}{10368}=0.0363619>V_4,
\end{align*}
which is a contradiction. Thus, we see that all the points of $\ga$ cannot lie on $x_2=\frac34$. Similarly, all the points of $\ga$ cannot lie on $x_1=\frac 12$. Recall that the Sierpi\'nski carpet has maximum symmetry with respect to the line $x_1=\frac 12$. As all the points of $\ga$ cannot lie on the line $x_1=\frac 12$, due to symmetry we can assume that the points of $\ga$ lie either on the three lines $x_1=\frac 16$, $x_1=\frac 56$ and $x_1=\frac 12$, or on the two lines $x_1=\frac 16$ and $x_1=\frac 56$.

Suppose $\ga$ contains points from the line $x_1=\frac 12$.
As $\ga$ cannot contain all the points from $x_1=\frac 12$, we can assume that $\ga$ contains two points, say $(\frac 12, b_1)$ and $(\frac 12, b_2)$ with $b_1<b_2$, from the line $x_1=\frac 12$ which are in the opposite sides of the centroid $(\frac 12, \frac 34)$, and the other two points, say $(\frac 16, a_1)$ and $(\frac 56, a_2)$, from the lines $x_1=\frac 16$ and $x_1=\frac 56$. Then, if $\ga$ does not contain any point from $J_3\uu J_4$, we have
\begin{align*}
&\int \min_{c\in \ga}\|x-c\|^2 dP\geq 2 \mathop{\int}\limits_{J_{31}\uu J_{33}}\|x-(\frac 16, \frac 23)\|^2dP=\frac{25}{768}=0.0325521>V_4,
\end{align*}
which leads to a contradiction. So, we can assume that $(\frac 16, a_1)\in J_3$ and $(\frac 56, a_2)\in J_4$. Suppose $\frac 23 \leq a_1, a_2\leq \frac 56$. Then, notice that $J_{31}\uu J_{33}\uu J_{321}\uu J_{323} \sci M((\frac 16, a_1)|\ga)$ and similar is the expression for the point $(\frac 56, a_2)$. Further, notice that $J_{11}\uu J_{12}\uu J_{14}\uu J_{21}\uu J_{22}\uu J_{23} \sci M((\frac 12, \frac 13)|\ga)$. Therefore, under the assumption $\frac 23 \leq a_1, a_2\leq \frac 56$, writing $A_1:=J_{31}\uu J_{33}\uu J_{321}\uu J_{323}$ and $A_2:=J_{11}\uu J_{12}\uu J_{14}$, we have the distortion error as
\begin{align*}
&\int \min_{c\in \ga}\|x-c\|^2 dP\geq 2 \Big(\mathop{\int}\limits_{A_1}\min_{\frac 23\leq b\leq \frac 56}\|x-(\frac 16, b)\|^2dP+\mathop{\int}\limits_{A_2}\min_{0\leq b\leq \frac 34}\|x-(\frac 12, b)\|^2dP\Big)\\
&=2 \Big(\frac{2051}{331776}+\frac{2021}{276480}\Big)=0.0269833>V_4,
\end{align*}
which leads to a contradiction. So, we can assume that $\frac 56<a_1, a_2\leq 1$. Then, we see that $J_1\uu J_2\sci M((\frac 12, b_1)|\ga)$ for $b_1=\frac 12$, and so the distortion error is
\[\int \min_{c\in \ga}\|x-c\|^2 dP\geq 2 \int_{J_1} \min_{0\leq b\leq \frac 34}\|x-(\frac 12, b)\|^2dP=\frac{13}{384}=0.0338542>V_4\]
which is a contradiction. All these contradictions arise due to our assumption that $\ga$ contains points from the line $x_1=\frac 12$. So, we can assume that $\ga$ cannot contain any point from the line $x_1=\frac 12$, i.e., we can assume that $\ga$ contains two points from the line $x_1=\frac 16$ and two points from the line $x_1=\frac 56$. Thus, we can take $\ga:=\set{(\frac 16, a_1), (\frac 16, b_1), (\frac 56, a_2), (\frac 56, b_2)}$ where $a_1\leq \frac 34\leq b_1$ and $a_2\leq \frac 34\leq b_2$. Notice that the Voronoi region of $(\frac 16, a_1)$ contains $J_1$ and the Voronoi region of $(\frac 56, a_2)$ contains $J_2$. If the Voronoi region of $(\frac 16, a_1)$ contains points from $J_3$, we must have $\frac 12(a_1+b_1)\geq \frac 23$ which yields $a_1\geq \frac 43-b_1\geq \frac 43 -\frac 34=\frac 7{12}$, and similarly if the Voronoi region of $(\frac 56, a_2)$ contains points from $J_4$, we must have $a_2\geq \frac 7{12}$. But, then
\begin{align*}
&\int \min_{c\in \ga}\|x-c\|^2 dP\geq  2\mathop{\int}\limits_{J_{1}}\|x-(\frac 16, \frac 7{12})\|^2dP+ 2\mathop{\int}\limits_{J_{33}\uu J_{34}}\|x-a(33, 34)\|^2dP\\
&=\frac{65}{1536}=0.0423177>V_4,
\end{align*}
which is a contradiction. So, we can assume that the Voronoi regions of $(\frac 16, a_1)$ and $(\frac 56, a_2)$ do not contain any point from $J_3\uu J_4$. Thus, we have
$(\frac 16, a_1)=a(1)=(\frac{1}{6},\frac{1}{4})$,  $(\frac 56, a_2)=a(2)=(\frac 56, \frac 14)$,  $(\frac 16, b_1)=a(3)=(\frac 1 6, \frac {11}{12})$,  and $(\frac 56, b_2)=a(4)=(\frac 56, \frac {11}{12})$, and the quantization error is  $V_4=\frac{7}{288}=0.0243056$. Thus, the proof of the proposition is complete.
\end{proof}

\begin{prop}
Let $n\geq 4$ and $\ga_n$ be an optimal set of $n$-means, and let $1\leq i\leq 4$. Then $\ga_n \II J_i \neq \es$, and $\ga_n\ii (J\setminus J_1\uu J_2\uu J_3\uu J_4)$ is an empty set.
\end{prop}
\begin{proof}
Let $\ga_n$ be an optimal set of $n$-means for $n\geq 4$. If $n=4$, the proposition is true by Proposition~\ref{prop3}. We now show that the proposition is true for $n\geq 5$. Consider the set of five points $\gb:=\set{(a(1), a(2), a(3, 3), a(3, 4), a(4)}$. The distortion error due to the set $\gb$ is given by
\[\int\min_{(a, b)\in\gb}\|x-(a, b)\|^2 dP=\frac{17}{864}=0.0196759.\]
Since $V_n$ is the quantization error for $n$-means for $n\geq 5$, we have $V_n\leq 0.0196759$. As described in the proof of Proposition~\ref{prop2}, we can assume that all the optimal points cannot lie in one side of the vertical line $x_1=\frac 12$ or in one side of the horizontal line $x_2=\frac 34$.

\end{proof}

\begin{note}
Let $\ga$ be an optimal set of $n$-means for some $n\geq 2$. Then, for $a\in \ga$, we have $a=a(\go)$, $a=a(\go1, \go3)$, or $a=a(\go2, \go4)$ for some $\go \in I^\ast$. Moreover, if $a\in \ga$, then $P$-almost surely $M(a|\ga)=J_\go$ if $a=a(\go)$, $M(a|\ga)=J_{\go1}\uu J_{\go3}$ if $a=a(\go1, \go3)$, and $M(a|\ga)=J_{\go2}\uu J_{\go4}$ if $a=a(\go2, \go4)$. For $\go \in I^\ast$,  $(i=1$ and $j=3)$, $(i=2$ and $j=4)$, or $(i=1, j=2)$ write
\begin{align}\label{eq2}
E(\go):=\mathop{\int}\limits_{J_\go}& \|x-a(\go)\|^2 dP, \te{ and }   E(\go i, \go j):=\mathop{\int}\limits_{J_{\go i}\uu J_{\go j}}\|x-a(\go i, \go j)\|^2 dP.
\end{align}
\end{note}
Let us now give the following lemma.

\begin{lemma} \label{lemma10}
For any $\go \in I^\ast$, let $E(\go)$, $E(\go1, \go3)$, $E(\go2, \go4)$, and $E(\go1, \go2)$ be defined by \eqref{eq2}. Then, $E(\go 1, \go3)=E(\go2, \go4)=\frac{31}{126} E(\go)$, $E(\go1, \go 2)=\frac{13}{84}E(\go)$, $E(\go1)=E(\go2)=\frac 1{72} E(\go)$, and $E(\go3)=E(\go4)=\frac 1{24} E(\go)$.
\end{lemma}

\begin{proof} By \eqref{eq1}, we have
\begin{align*}
&E(\go1, \go3)=\mathop{\int}\limits_{J_{\go1}\uu J_{\go3}}\|x-a(\go1, \go3)\|^2 dP=\mathop{\int}\limits_{J_{\go1}}\|x-a(\go1, \go3)\|^2 dP+\mathop{\int}\limits_{J_{\go3}}\|x-a(\go1, \go3)\|^2 dP\\
&=p_{\go1} (s_{\go1}^2V+\|a(\go1)-a(\go1, \go3)\|^2)+p_{\go3} (s_{\go3}^2V+\|a(\go3)-a(\go1, \go3)\|^2).
\end{align*}
Notice that
\[a(\go1, \go3)=\frac{1}{p_{\go1}+p_{\go3}}\Big (p_{\go1}S_{\go1}(\frac12, \frac 34)+p_{\go3}S_{\go3}(\frac12, \frac 34)\Big)=\frac{1}{\frac 18+\frac 38}\Big(\frac 18 S_{\go1}(\frac12, \frac 34)+\frac 38S_{\go3}(\frac12, \frac 34)\Big),\]
which implies $a(\go1, \go3)=\frac 14 S_{\go1}(\frac12, \frac 34)+\frac 34 S_{\go3}(\frac12, \frac 34)$. Thus, we have
\begin{align*}
&\|a(\go1)-a(\go1, \go3)\|^2=\|S_{\go1}(\frac 12, \frac 34)-\frac 14 S_{\go1}(\frac12, \frac 34)-\frac 34 S_{\go3}(\frac12, \frac 34)\|^2=\frac 9{16} s_\go^2 \|(0, \frac 23)\|^2=\frac 14 s_\go^2,
\end{align*}
and similarly, $\|a(\go3)-a(\go1, \go3)\|^2=\frac 1{16} s_\go^2\|(0, \frac 23)\|^2 =\frac 1{36} s_\go^2$. Thus, we obtain,
\begin{align*}
&E(\go1, \go3)=p_{\go1} (s_{\go1}^2V+\frac 1{4} s_\go^2)+p_{\go3} (s_{\go3}^2V+\frac 1{36} s_\go^2)=p_\go s_\go^2 V(p_1s_1^2+p_3s_3^2)+p_\go s_\go^2(\frac 14 p_1+\frac 1{36}p_3)\\
&= p_\go s_\go^2 V(\frac 1{18}+\frac1{24}\frac 1 V)=\frac {31}{126} E(\go),
\end{align*}
and similarly, we can prove the rest of the lemma. Thus, the proof of the lemma is complete.
\end{proof}

\begin{remark}
From the above lemma it follows that $E(\go1, \go3)=E(\go2, \go4)>E(\go1, \go2)>E(\go3)=E(\go4)>E(\go1)=E(\go2)$.
\end{remark}
The following lemma gives some important properties about the distortion error.

\begin{lemma} \label{lemma11}

Let $\go, \gt \in I^\ast$. Then

$(i)$ $E(\go)> E(\gt)$ if and only if $E(\go1, \go3)+E(\go2, \go4)+E(\gt)< E(\go)+E(\gt1, \gt3)+E(\gt2, \gt4)$;

$(ii)$ $E(\go)> E(\gt 1, \gt 3)(=E(\gt2, \gt4))$ if and only if $E(\go1, \go3)+E(\go2, \go 4)+E(\gt 1, \gt 3)+E(\gt 2, \gt 4)< E(\go)+E(\gt 1, \gt2)+E(\gt 3)+E(\gt4)$;

$(iii)$  $E(\go1, \go3)(=E(\go 2, \go 4))> E(\gt 1, \gt 3)(=E(\gt2, \gt4))$ if and only if $E(\go1, \go2)+E(\go3)+E(\go 4)+E(\gt 1, \gt 3)+E(\gt 2, \gt4)< E(\go1, \go3)+E(\go2, \go4)+E(\gt 1, \gt2)+E(\gt 3)+E(\gt4)$;

$(iv)$   $E(\go1, \go3)(=E(\go 2, \go 4))> E(\gt)$ if and only if $E(\go1, \go2)+E(\go3)+E(\go 4)+E(\gt)< E(\go1, \go3)+E(\go2, \go4)+E(\gt 1, \gt3)+E(\gt 2, \gt 4)$;

$(v)$ $E(\go1, \go2)> E(\gt)$ if and only if $E(\go1)+E(\go2)+E(\gt)< E(\go1, \go2)+E(\gt1, \gt3)+E(\gt2, \gt4)$;

$(vi)$ $E(\go1, \go2)> E(\gt 1, \gt 3)(=E(\gt2, \gt4))$ if and only if $E(\go1)+E(\go2)+E(\gt 1, \gt 3)+E(\gt 2, \gt 4)< E(\go1, \go2)+E(\gt 1, \gt2)+E(\gt 3)+E(\gt4)$;

$(vii)$ $E(\go1, \go2)> E(\gt 1, \gt 2)$ if and only if $E(\go1)+E(\go2)+E(\gt 1, \gt 2)< E(\go1, \go2)+E(\gt 1)+E(\gt2)$;

$(viii)$ $E(\go)> E(\gt1,\gt2)$ if and only if $E(\go1, \go3)+E(\go2, \go4)+E(\gt1, \gt 2)< E(\go)+E(\gt1)+E(\gt2)$.
\end{lemma}

\begin{proof} Let us first prove $(iii)$. Using Lemma~\ref{lemma10}, we see that
\begin{align*}
LHS&=E(\go1, \go2)+E(\go3)+E(\go 4)+E(\gt 1, \gt 3)+E(\gt 2, \gt4)=\frac {5}{21} E(\go)+\frac{31}{63}E(\gt),\\
RHS&= E(\go1, \go3)+E(\go2, \go4)+E(\gt 1, \gt2)+E(\gt 3)+E(\gt4)=\frac{31}{63} E(\go)+\frac {5}{21} E(\gt).
\end{align*}
Thus, $LHS< RHS$ if and only if $\frac {5}{21} E(\go)+\frac{31}{63}E(\gt)< \frac{31}{63} E(\go)+\frac {5}{21} E(\gt)$, which yields $E(\go)>E(\gt)$, i.e., $E(\go1, \go3)>E(\gt1, \gt3)$. Thus $(iii)$ is proved. The other parts of the lemma can similarly be proved. Thus, the lemma follows.
\end{proof}

%

\begin{figure}
\centerline{\includegraphics[width=7.5 in, height=7 in]{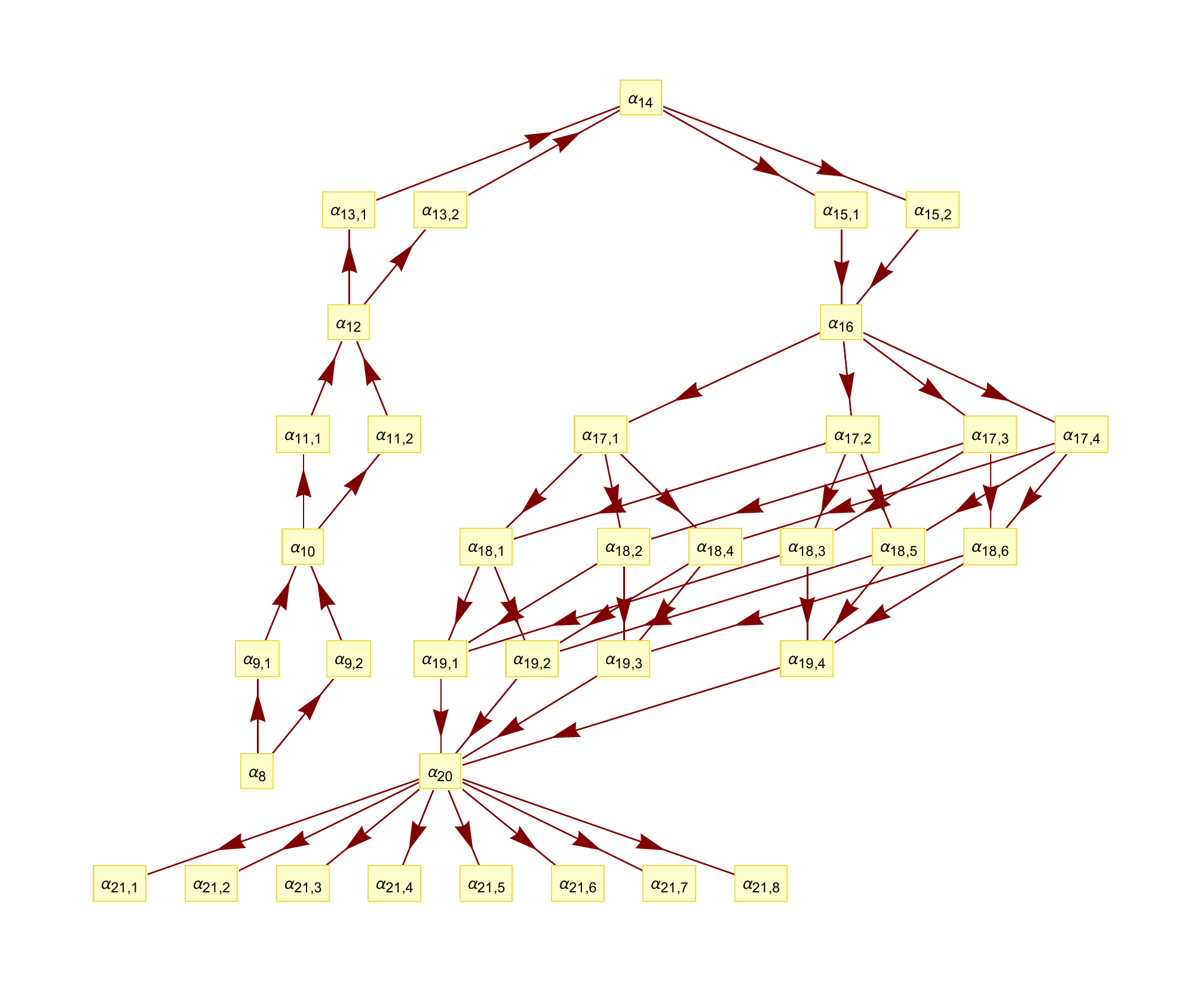}}
\caption{Tree diagram of the optimal sets from $\ga_8$ to $\ga_{21}$.} \label{Fig2}
\end{figure}

In the following theorem, we give the induction formula to determine the optimal sets of $n$-means for any $n\geq 2$.

\begin{theorem} \label{Th1} For any $n\geq 2$, let $\ga_n:=\set{a(i) : 1\leq  i\leq n}$ be an optimal set of $n$-means, i.e., $\alpha_n \in\C C_n:= \mathcal{C}_n(P)$. For $\go \in I^\ast$, let $E(\go)$, $E(\go1, \go3)$ and $E(\go 2, \go4)$ be defined by \eqref{eq2}. Set
\[\tilde  E(a(i)):=\left\{\begin{array} {ll}
E(\go) \te{ if } a(i)=a(\go) \te{ for some }  \go \in I^\ast, \\
E(\go k, \go \ell) \te{ if } a(i)=a(\go k, \go \ell) \te { for some } \go \in I^\ast,
\end{array} \right.
\]
where $(k=1, \ell=3)$, or $(k=2, \ell=4)$, or $(k=1, \ell=2)$,
and $W(\ga_n):=\set{a(j)  : a(j) \in \ga_n \te{ and } \tilde E(a(j))\geq \tilde E(a(i)) \te{ for all } 1\leq i\leq n}$. Take any $a(j) \in W(\ga_n)$, and write
\[\ga_{n+1}(a(j)):=\left\{\begin{array}{ll}
(\ga_n\setminus \set{a(j)})\uu \set{a(\go1, \go3), a(\go2, \go4)} \te{ if } a(j)=a(\go), &\\
(\ga_n \setminus \set{a(\go1, \go3), a(\go2, \go4)})\uu \set{a(\go 1, \go2), a(\go 3), a(\go 4)} & \\
\qquad \qquad \te{ if } a(j)=a(\go 1, \go 3) \te{ or } a(\go2, \go 4), &\\
(\ga_n\setminus \set{a(j)})\uu \set{a(\go1),  a(\go2)} \te{ if } a(j)=a(\go1, \go2),
\end{array}\right.
\]
Then $\ga_{n+1}(a(j))$ is an optimal set of $(n+1)$-means, and the number
of such sets is given by
\[\te{card}\Big(\UU_{\alpha_n \in \C{C}_n}\{\alpha_{n+1}(a(j)) : a(j) \in W(\ga_n)\}\Big).\]
\end{theorem}

\begin{proof}
By Proposition~\ref{prop1}, Proposition~\ref{prop2} and Proposition~\ref{prop3}, we know that the optimal sets of two-, three-, and four-means are respectively $\{a(1,3), a(2, 4)\}$,
$\{a(1, 2), a(3), a(4)\}$, and $\set{a(1), a(2), a(3), a(4)}$. Notice that by Lemma~\ref{lemma10}, we know $E(1, 3)\geq E(2, 4)$, and $E(1, 2)\geq E(3)=E(4)$. Thus, the lemma is true for $n=2$ and $n=3$.
 For any $n\geq 3$, let us now assume that $\alpha_n$ is an optimal
set of $n$-means. Let $\ga_n:=\set{a(i) : 1\leq i\leq n}$. Let $\tilde  E(a(i))$ and $W(\ga_n)$ be defined as in the hypothesis. If $a(j) \not \in W(\ga_n)$, i.e., if  $a(j) \in \ga_n\setminus W(\ga_n)$, then by Lemma~\ref{lemma11}, the error
\[\sum_{a(i)\in (\ga_n\setminus \set{a(j)})} \tilde E(a(i))+E(\go1, \go3)+E(\go2, \go4) \te{ if } a(j)=a(\go),\]
\[\sum_{a(i)\in (\ga_n\setminus \set{a(\go1, \go3), \, a(\go2, \go4)})} \tilde E(a(i))+E(\go1, \go2)+E(\go3)+E(\go4) \te{ if } a(j)=a(\go1, \go3) \te{ or } a(\go 2, \go 4),\]
\[\sum_{a(i)\in (\ga_n\setminus \set{a(j)})} \tilde E(a(i))+E(\go1)+E(\go2) \te{ if } a(j)=a(\go1, \go2),\]
obtained in this case is strictly greater than the corresponding error obtained in the case when $a(j)\in W(\ga_n)$. Hence for any $a(j) \in W(\ga_n)$, the set $\ga_{n+1}(a(j))$, where
\[\ga_{n+1}(a(j)):=\left\{\begin{array}{ll}
(\ga_n\setminus \set{a(j)})\uu \set{a(\go1, \go3), a(\go2, \go4)} \te{ if } a(j)=a(\go), &\\
(\ga_n \setminus \set{a(\go1, \go3), a(\go2, \go4)})\uu \set{a(\go 1, \go2), a(\go 3), a(\go 4)} & \\
\qquad \qquad \te{ if } a(j)=a(\go 1, \go 3) \te{ or } a(\go2, \go 4), &\\
(\ga_n\setminus \set{a(j)})\uu \set{a(\go1),  a(\go2)} \te{ if } a(j)=a(\go1, \go2),
\end{array}\right.
\] is an optimal set of $(n+1)$-means, and the number
of such sets is
\[\te{card}\Big(\UU_{\alpha_n \in \C{C}_n}\{\alpha_{n+1}(a(j)) : a(j) \in W(\ga_n)\}\Big).\]
Thus the proof of the theorem is complete (also see Note~\ref{note10}).
\end{proof}

\begin{remark}

Once an optimal set of $n$-means is known, by using \eqref{eq1}, the corresponding quantization error can easily be calculated.
\end{remark}

\begin{remark}
By Theorem~\ref{Th1}, we note that  to obtain an optimal set of $(n+1)$-means one needs to know an optimal set of $n$-means. We conjecture that unlike the uniform probability distribution, i.e., when the probability measures on the basic rectangles at each level of the Sierpi\'nski carpet construction are equal, for the nonuniform  probability distribution considered in this paper, to obtain the optimal sets of $n$-means a closed formula cannot be obtained.
\end{remark}

Running the induction formula given by Theorem~\ref{Th1} in computer algorithm, we obtain some results and observations about the optimal sets of $n$-means, which are given in the following section.

\section{Some results and observations}

First, we explain about some notations that we are going to use in this section. Recall that the optimal set of one-mean consists of the expected value of the random variable $X$, and the corresponding quantization error is its variance. Let $\ga_n$ be an optimal set of $n$-means, i.e., $\ga_n \in \C C_n$, and then for any $a\in \ga_n$, we have  $a=a(\go)$, or $a=a(\go i, \go j)$ for some $\go \in I^\ast$, where $(i=1, j=3)$, $(i=2, j=4)$, or $(i=1, j=2)$. For $\go=\go_1\go_2\cdots\go_k\in I^k$, $k\geq 1$, in the sequel, we will identify the elements $a(\go)$ and $a(\go i, \go j)$ by the sets $\set{\set{\go_1, \go_2, \cdots, \go_k}}$ and $\set{\set{\go_1, \go_2, \cdots, \go_k, i}, \set{\go_1, \go_2, \cdots, \go_k, j}}$, respectively. Thus, we can write
\begin{align*} &\ga_2= \set{\set{\set{1}, \set{3}}, \set{\set{2} ,\set{4}}}, \,  \ga_3=\set{\set{\set{1}, \set {2}}, \set{\set{3}}, \set{\set{4}}}, \\
 &\ga_4=\set{\set{\set{1}}, \set{\set{2}}, \set{\set{3}}, \set{\set{4}}},
\end{align*}
and so on. For any $n\geq 2$, if $\te{card}(\C C_n)=k$, we write
\[\C C_n=\left\{\begin{array}{ccc}
\set{\ga_{n, 1}, \ga_{n, 2}, \cdots, \ga_{n, k}} & \te{ if } k\geq 2,\\
\set{\ga_{n}} & \te{ if } k=1.
\end{array}\right.
\]
If $\te{card}(\C C_n)=k$ and  $\te{card}(\C C_{n+1})=m$, then either $1\leq k\leq m$, or $1\leq m\leq k$ (see Table~\ref{tab1}). Moreover, by Theorem~\ref{Th1}, an optimal set at stage $n$ can contribute multiple distinct optimal sets at stage $n+1$, and multiple distinct optimal sets at stage $n$ can contribute one common optimal set at stage $n+1$; for example from Table~\ref{tab1}, one can see that the number of $\ga_{21}=8$, the number of $\ga_{22}=28$, the number of $\ga_{23}=56$, the number of $\ga_{24}=70$, and the number of $\ga_{25}=56$.

 By $\ga_{n, i} \rightarrow \ga_{n+1, j}$, it is meant that the optimal set $\ga_{n+1, j}$ at stage $n+1$ is obtained from the optimal set $\ga_{n, i}$ at stage $n$, similar is the meaning for the notations $\ga_n\rightarrow \ga_{n+1, j}$, or $\ga_{n, i} \rightarrow \ga_{n+1}$, for example from Figure~\ref{Fig2}:
 \begin{align*} &\left\{\alpha _{16}\to \alpha _{17,1},\alpha _{16}\to \alpha _{17,2},\alpha _{16}\to \alpha _{17,3},\alpha _{16}\to \alpha _{17,4}\right\};\\
 &\{\left\{\alpha _{17,1}\to \alpha _{18,1},\alpha _{17,1}\to \alpha _{18,2},\alpha _{17,1}\to \alpha _{18,4}\right\}, \left\{\alpha _{17,2}\to \alpha _{18,1},\alpha _{17,2}\to \alpha _{18,3},\alpha _{17,2}\to \alpha _{18,5}\right\}, \\
 & \left\{\alpha _{17,3}\to \alpha _{18,2},\alpha _{17,3}\to \alpha _{18,3},\alpha _{17,3}\to \alpha _{18,6}\right\},\left\{\alpha _{17,4}\to \alpha _{18,4},\alpha _{17,4}\to \alpha _{18,5},\alpha _{17,4}\to \alpha _{18,6}\right\}\};\\
 & \{\left\{\alpha _{18,1}\to \alpha _{19,1},\alpha _{18,1}\to \alpha _{19,2}\right\},\left\{\alpha _{18,2}\to \alpha _{19,1},\alpha _{18,2}\to \alpha _{19,3}\right\}, \left\{\alpha _{18,3}\to \alpha _{19,1},\alpha _{18,3}\to \alpha _{19,4}\right\}, \\
 & \left\{\alpha _{18,4}\to \alpha _{19,2},\alpha _{18,4}\to \alpha _{19,3}\right\}, \left.\left\{\alpha _{18,5}\to \alpha _{19,2},\alpha _{18,5}\to \alpha _{19,4}\right\},\left\{\alpha _{18,6}\to \alpha _{19,3},\alpha _{18,6}\to \alpha _{19,4}\right\}\right\};\\
 & \left\{\alpha _{19,1}\to \alpha _{20},\alpha _{19,2}\to \alpha _{20},\alpha _{19,3}\to \alpha _{20},\alpha _{19,4}\to \alpha _{20}\right\}.
\end{align*}

Moreover, one can see that
\begin{align*}
\ga_8&=\Big\{\{\{1,1\},\{1,3\}\},\{\{1,2\},\{1,4\}\},\{\{2,1\},\{2,3\}\},\{\{2,2\},\{2,4\}\},\{\{3,1\},\{3,3\}\},\\
 & \{\{3,2\},\{3,4\}\},\{\{4,1\},\{4,3\}\},\{\{4,2\},\{4,4\}\Big\}  \te{ with }  V _8= \frac{31}{2592}=0.0119599;\\
\ga_{9,1}&= \Big \{\{\{3,3\}\},\{\{3,4\}\},\{\{1,1\},\{1,3\}\},\{\{1,2\},\{1,4\}\},\{\{2,1\},\{2,3\}\},\{\{2,2\},\{2,4\}\},\\
& \{\{3,1\},\{3,2\}\},\{\{4,1\},\{4,3\}\},\{\{4,2\},\{4,4\}\}\Big \},\\
\ga_{9,2}&= \Big \{\{\{4,3\}\},\{\{4,4\}\},\{\{1,1\},\{1,3\}\},\{\{1,2\},\{1,4\}\},\{\{2,1\},\{2,3\}\},\{\{2,2\},\{2,4\}\},\\
& \{\{3,1\},\{3,3\}\}, \{\{3,2\},\{3,4\}\},\{\{4,1\},\{4,2\}\}\Big \} \te{ with }  V_9= \frac{25}{2592}=0.00964506;\\
\ga_{10} &=\Big\{\{\{3,3\}\},\{\{3,4\}\},\{\{4,3\}\},\{\{4,4\}\},\{\{1,1\},\{1,3\}\},\{\{1,2\},\{1,4\}\},\{\{2,1\},\{2,3\}\},\\
& \{\{2,2\},\{2,4\}\},\{\{3,1\},\{3,2\}\},\{\{4,1\},\{4,2\}\}\Big \} \te{ with } V_{10}=\frac{19}{2592}=0.00733025,\end{align*}
and so on.
\begin{table}
\begin{center}
\begin{tabular}{ |c|c||c|c|| c|c||c|c|c||c|c||c|c}
 \hline
$n$ & $\te{card}(\C C_n) $ & $n$ & $\te{card}(\C C_n) $  & $n$ & $\te{card}(\C C_n)  $  & $n$ & $\te{card}(\C C_n)$   & $n$ & $\te{card}(\C C_n)$ & $n$ & $\te{card}(\C C_n)$\\
 \hline
5 & 2 & 18 &  6 &  31& 4 & 44 & 70& 57 & 8& 70 & 6 \\6 & 1 &  19 & 4 & 32 & 1 & 45& 56 & 58 & 28& 71 & 4\\7 & 2 & 20  & 1 & 33&  4& 46 & 28 & 59 & 56& 72 & 1\\8 & 1 &  21 &  8 & 34 & 6& 47& 8 & 60&70 & 73 & 24\\9 & 2 & 22&  28 & 35 &  4 & 48 & 1 & 61 & 56 & 74 & 276 \\10 & 1 & 23 & 56 & 36 & 1 & 49 & 8& 62 & 28& 75 & 2024 \\11 & 2 & 24 & 70 & 37 & 4& 50  & 28& 63 & 8 & 76 & 10626 \\12 & 1 & 25 & 56 & 38 & 6 & 51 & 56 & 64& 1& 77 & 42504\\13 & 2 & 26 & 28 & 39& 4 & 52 &70 & 65 & 4& 78 & 134596 \\14 & 1 & 27& 8 & 40 & 1& 53 & 56& 66  & 6 & 79 & 346104\\15 & 2 & 28&  1 & 41 &8 & 54 & 28& 67 & 4 & 80 & 735471 \\16 & 1 &29 & 4 & 42 & 28 & 55 &8& 68 &1& 81 & 1307504\\17 & 4 &30 & 6 & 43&  56 & 56 & 1 & 69 & 4 & 82 & 1961256\\
 \hline
\end{tabular}
 \end{center}
 \vspace{ 0.5 in}
 \caption{Number of $\ga_n$ in the range $5\leq n\leq 82$.}
    \label{tab1}
\end{table}

\begin{note} \label{note10} Notice that there is only one optimal set of $n$-means for $n=72$. By the notations used in Theorem~\ref{Th1}, we can write $\ga_{72}=\set{a(i) : 1\leq i\leq 72}$. Then,
\begin{align*}W(\ga_{72})&= \set{\set{\set{1, 3, 3}}, \set{\set{1, 3, 4}}, \set{\set{1, 4, 3}}, \set{\set{1, 4, 4}}, \set{\set{2, 3, 3}}, \set{\set{2,
    3, 4}}, \set{\set{2, 4, 3}}, \\
    & \set{\set{2, 4, 4}}, \set{\set{3, 1, 3}}, \set{\set{3, 1, 4}}, \set{\set{3,
   2, 3}}, \set{\set{3, 2, 4}}, \set{\set{3, 3, 1}}, \set{\set{3, 3, 2}}, \\
   & \set{\set{3, 4, 1}}, \set{\set{3, 4,
    2}}, \set{\set{4, 1, 3}}, \set{\set{4, 1, 4}}, \set{\set{4, 2, 3}}, \set{\set{4, 2, 4}}, \set{\set{4, 3,
   1}}, \set{\set{4, 3, 2}}, \\
   & \set{\set{4, 4, 1}}, \set{\set{4, 4, 2}}}.
   \end{align*}
Since $\te{card}(W(\ga_{72}))=24$, by the theorem, we have
$\te{card}(\C C_{73})= \binom{24}{1}=24, \, \te{card}(\C C_{74})= \binom{24}{2}=276, \, \te{card}(\C C_{75})= \binom{24}{3}=2024, \, \te{card}(\C C_{76})= \binom{24}{4}=10626$, etc., for details see Table~\ref{tab1}.
\end{note}

Let us now conclude the paper with the following remark:

\begin {remark}

Consider a set of four contractive affine transformations $S_{(i, j)}$ on $\D R^2$, such that
 $S_{(1,1)}(x_1, x_2)=(\frac 1 4 x_1, \frac 1 4 x_2)$, $S_{(2,1)}(x_1, x_2)=(\frac 12 x_1+\frac 12, \frac 1 4 x_2)$, $S_{(1,2)}(x_1, x_2)=(\frac 14 x_1, \frac 12 x_2+\frac 12)$, and $S_{(2,2)}(x_1, x_2)=(\frac 12 x_1+\frac 12,  \frac 12 x_2+\frac 12)$ for all $(x_1, x_2) \in \D R^2$. Let $S$ be the limit set of these contractive mappings. Then, $S$ is called the Sierpi\'nski carpet generated by $S_{(i, j)}$ for all $1\leq i, j\leq 2$.
 Let $P$ be the Borel probability measure on $\D R^2$ such that $P=\frac{1}{16} P\circ S_{(1,1)}^{-1}+\frac 3{16} P\circ S_{(2,1)}^{-1} +\frac 3{16} P\circ S_{(1,2)}^{-1} +\frac {9}{16} P\circ S_{(2,2)}^{-1}$. Then, $P$ has support the Siepi\'nski carpet $S$. For this probability measure, the optimal sets of $n$-means and the $n$th quantization errors are not known yet for all $n\geq 2$.

\end{remark}

\end{document}